\newif\ifcccg
\cccgfalse

\ifcccg
  \documentclass{cccg14}
  \usepackage{amsmath,amssymb,amsfonts}
\else
  \documentclass[a4paper]{article}
  \usepackage{fullpage}
  \usepackage{amsthm,amsmath,amssymb,amsfonts}

  \newtheorem{theorem}{Theorem}

\fi

\usepackage{url,color,booktabs,verbatim,xspace}

\usepackage{graphicx}
\usepackage{epstopdf} 

\newcommand{\IGNORE}[1]{}

\newcommand{\Sei}{Sei Sh\=onagon\protect\xspace}
\newcommand{\game}{\Sei Chie no Ita\protect\xspace}

\title{The Convex Configurations of ``\game'' and Other Dissection Puzzles}
\author{
Eli Fox-Epstein\thanks{Brown University, USA, {\tt ef@cs.brown.edu}}
\and
Ryuhei Uehara\thanks{School of Information Science, JAIST, Japan, {\tt uehara@jaist.ac.jp}}
}
\date{}
\index{Fox-Epstein, Eli}
\index{Uehara, Ryuhei}
\begin{document}
\maketitle

\begin{abstract}
The \emph{tangram} and \emph{\game} are popular dissection puzzles consisting of seven pieces. Each puzzle can be formed by identifying edges from sixteen identical right isosceles triangles.
It is known that the tangram can form 13 convex polygons.
We show that \game can form 16 convex polygons, propose a new puzzle that can form 19,
no 7 piece puzzle can form 20,
and 11 pieces are necessary and sufficient to form
all 20 polygons formable by 16 identical isosceles right triangles. 
Finally, we examine the number of convex polygons formable by different quantities of these triangles.
\end{abstract}

\section{Introduction}

A \emph{dissection puzzle} is a game where one must decide
whether a given set of polygons can be placed in the plane
in such a way that their union is a given target polygon.
Rotation and reflection are allowed but scaling is not,
and all polygons must be internally disjoint.
Formally, a set of polygons $S$ can \emph{form} a polygon $P$
if there is an isomorphism up to rotation and reflection between
a partition of $P$ and the polygons of $S$ (i.e. a bijection $f(\cdot)$ from a partition of $P$ to $S$ such that $x$ and $f(x)$ are congruent for all $x$).

The \emph{tangram} is a set of polygons consisting of a square 
cut by straight incisions into different-sized pieces.
See the left diagram in \figurename~\ref{fig:diagrams}.
Of anonymous origin, the first known reference in literature
is from 1813 in China \cite{Slocum}.
The tangram has grown to be extremely popular throughout the world and now has
over 2000 dissection and related puzzles \cite{Slocum,Gardner}.

\begin{figure}[h]
  \centering
  \includegraphics[width=0.35\textwidth]{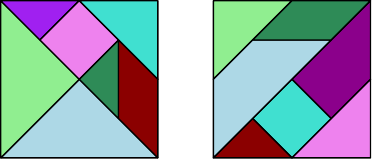}
  \caption{Left: the tangram in square configuration.
           Right: \game pieces in square configuration.}
  \label{fig:diagrams}
\end{figure}

\begin{figure}[t]
  \centering
  \includegraphics[width=0.4\textwidth,clip=true,trim=0 20 100 20]{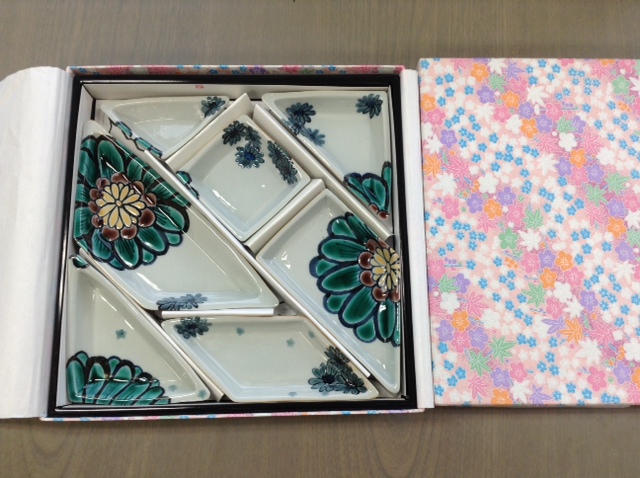}
  \caption{A set of plates in the form of \game pieces, crafted by Tomomi Takeda in Kanazawa, Japan.}
  \label{fig:pottery}
\end{figure}
\begin{figure}[t]
 \centering
  \includegraphics[width=0.2\textwidth]{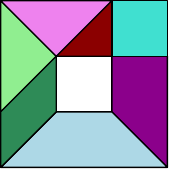}
  \caption{A typical \game layout as a square configuration with a hole missing.}
  \label{fig:hole}
\end{figure}

There is a similar but less famous Japanese set of puzzle pieces called \emph{\game}.
\Sei was a courtier and famous novelist in Japan, but 
there is no evidence that the puzzle existed a millennium ago during her lifetime.
\emph{Chie no ita} means \emph{wisdom plates}, which refers to the physical puzzle.
It is said that the puzzle is named after \Sei's wisdom.
Historically, the \game first appeared in literature in 1742 \cite{Slocum}.
Even in Japan, the tangram is more popular than \game, though
\game is common enough to have been made into ceramic dinner plates (see e.g. \figurename~\ref{fig:pottery}, \cite{pottery}), and in puzzle communities it is admired for being able to form some more interesting shapes that the tangram cannot, such as a square configuration with 
a hole missing (\figurename~\ref{fig:hole}).

Wang and Hsiung considered the number of possible convex
  (filled) polygons formed by the tangram \cite{WH42}.
They first noted that, given sixteen identical isosceles right triangles,
 one can create the tangram pieces by gluing some edges together.
Consequently, the tangram pieces can only form a subset of the convex polygons that sixteen idential isosceles right triangles can form.
  Embedded in the proof of their main theorem, Wang and Hsiung \cite{WH42} demonstrate
  that sixteen identical isosceles right triangles can form exactly 20 convex polygons.
These 20 are illustrated in \figurename~\ref{fig:20}.
The tangram can realize thirteen of them.

\begin{figure*}[h]
  \centering
  \includegraphics[height=0.19\textheight]{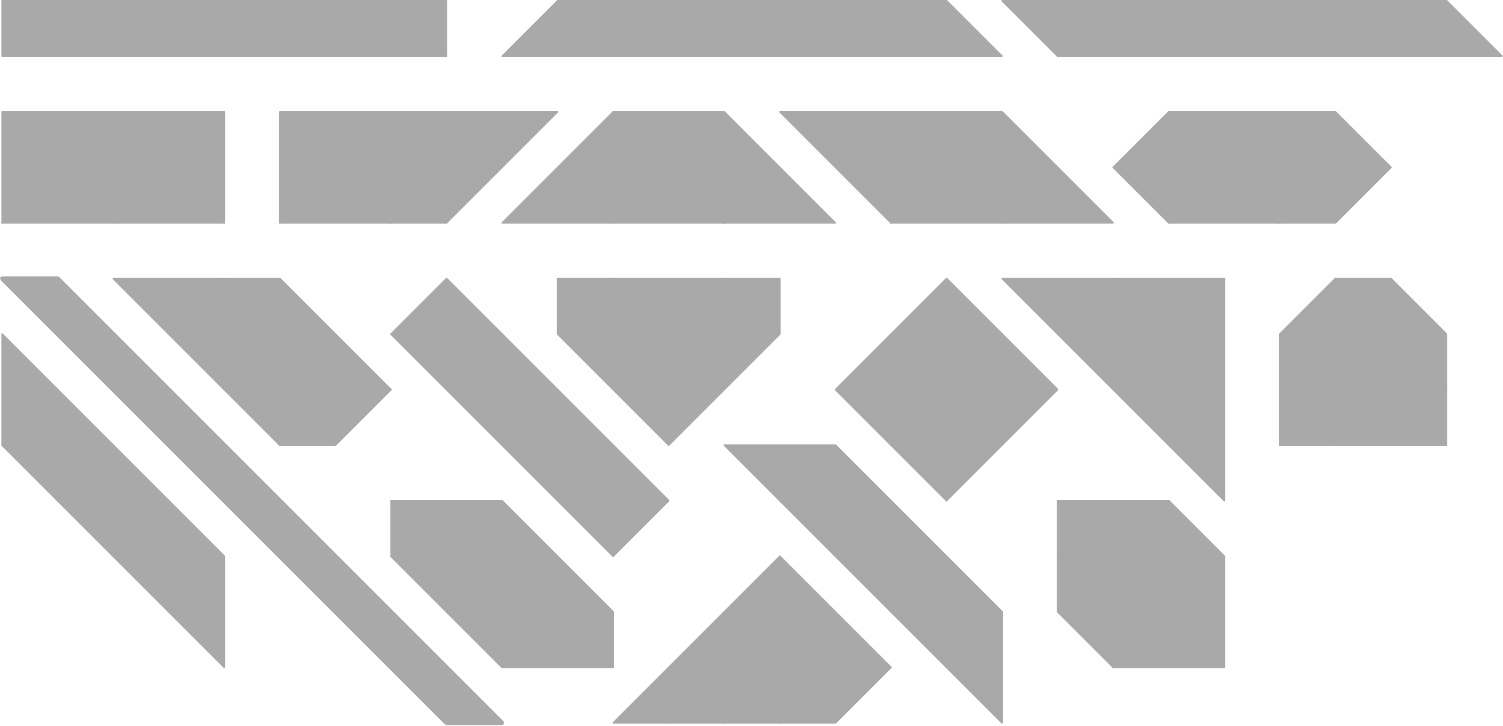}
  \caption{All 20 potential convex polygons.}
  \label{fig:20}
\end{figure*}

It is quite natural to ask how many of these convex polygons 
the \game pieces can form.
We first show that \game achieves sixteen.
Therefore, in a sense, we can conclude that \game is more expressive than the tangram:
while both the tangram and \game contain seven pieces made from sixteen identical
isosceles right triangles, \game can form more convex polygons than the tangram.
(Also, recall that the \game configuration in \figurename~\ref{fig:hole} is impossible with the tangram.)

One might next wonder if this can be improved with different shapes.
We demonstrate a set of seven pieces that can form nineteen convex polygons
among the 20 candidates, and that to realize them all, it
is necessary to have at least eleven shapes, which is sufficient.
Throughout, all triangles mentioned are identical isosceles right triangles with side lengths $1$, $1$, and $\sqrt 2$.

\section{The \game puzzle}

\begin{theorem}
  The \game puzzle pieces can be rearranged into exactly
  sixteen distinct convex polygons up to reflection and rotation.
  \label{thm:main}
\end{theorem}
\begin{proof}
We first notice that the seven puzzle pieces can be decomposed into
sixteen identical right isosceles triangles, just like the tangram.

We make use of two important results from Wang and Hsiung \cite{WH42}.
First, there are only 20 candidate convex polygons that we need to consider;
  and second, in any formable convex polygon,
  the bases of the sixteen triangles can be pairwise collinear,
  parallel, or perpendicular (\cite{WH42}, Lemma 1).
This means we only need to consider configurations that could be embedded with triangle and target polygon vertices on integer coordinates.

Sixteen convex polygons are filled as illustrated in \figurename~\ref{fig:example}.
The remaining four polygons cannot be formed since they are too thin.
More precisely, the largest trapezoid puzzle piece of area 2 has a base of length 3.
Under the four rotations we need to consider, the base of the trapezoid does not
  fit into the target polygon.
\end{proof}
\begin{figure*}[h]
  \centering
  \includegraphics[height=0.19\textheight]{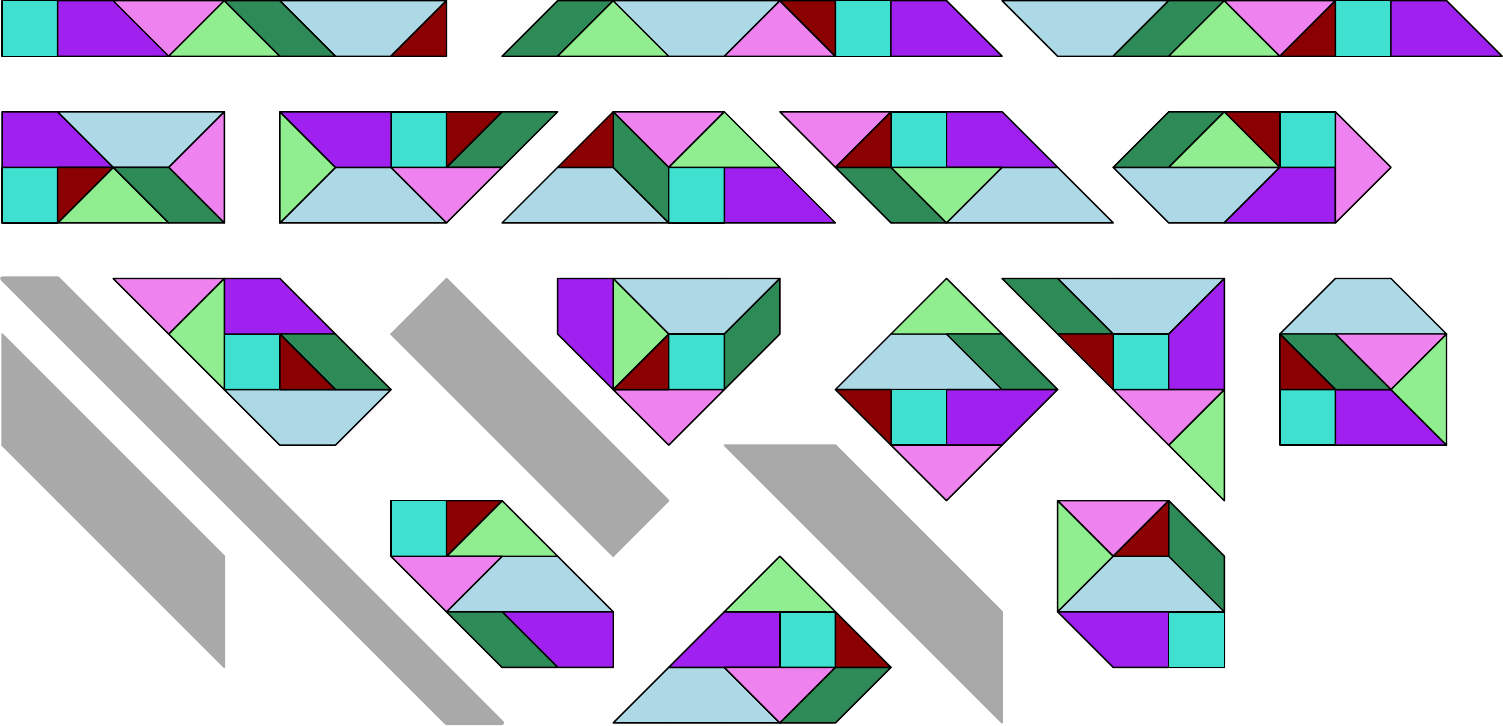}
  \caption{The sixteen convex polygons that can be formed by \game.}
  \label{fig:example}
\end{figure*}

\section{An optimal seven piece puzzle}

Although \game is more expressive than the tangram, 
\game{} is not the optimal set of seven pieces 
if one wishes to form as many convex polygons as possible.

\begin{theorem}
  There is a set of seven polygons composed from sixteen identical right isosceles triangles
  that can form nineteen distinct convex polygons.
  Furthermore, no set of seven polygons composed of sixteen identical right isosceles triangles can form 20 distinct convex polygons.
\end{theorem}
\begin{proof}
The set of seven polygons that can form nineteen distinct convex polygons and
its formations are depicted in \figurename~\ref{fig:19}.
Theorem~\ref{thm:eleven-pieces} implies that no seven-piece puzzle can form all 20 convex polygons.
\end{proof}
\begin{figure*}[h]
  \centering
  \includegraphics[height=0.19\textheight]{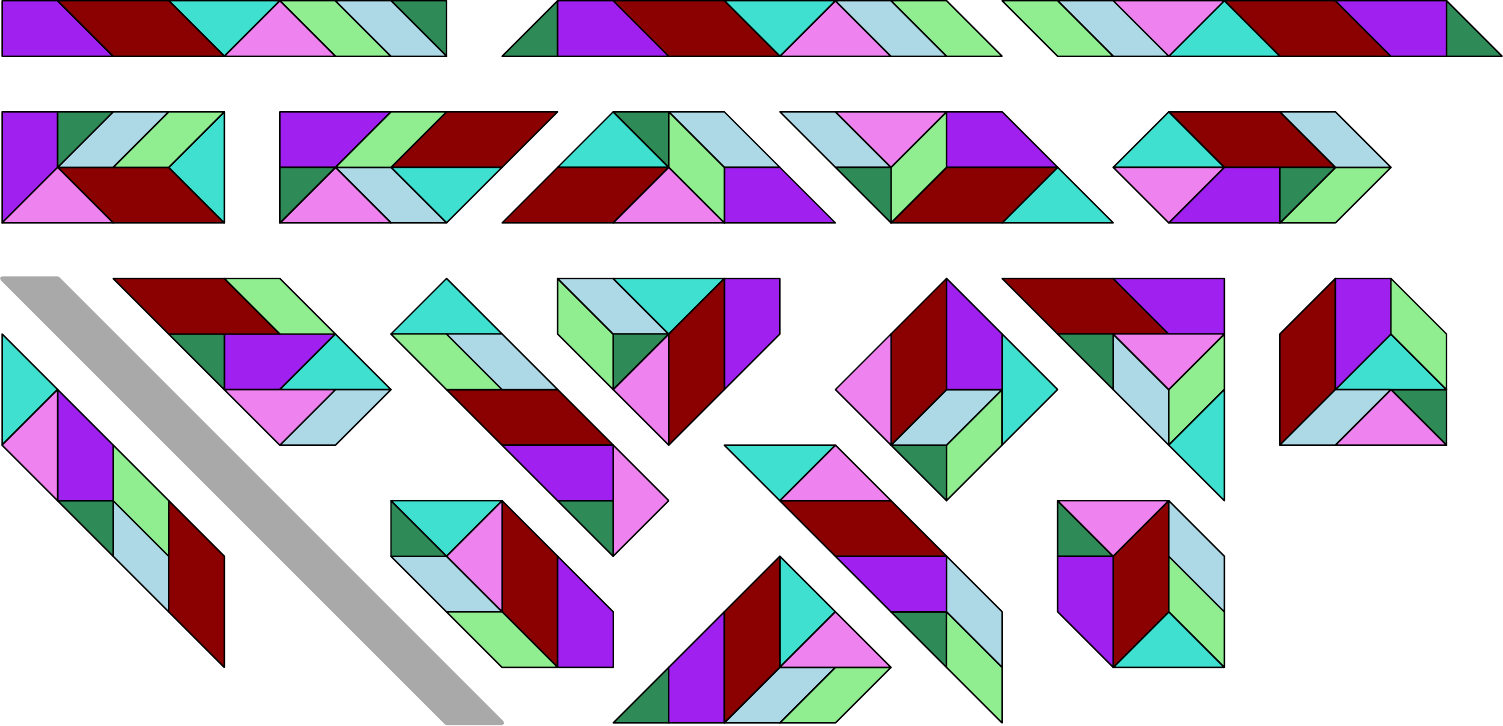}
  \caption{Seven pieces forming nineteen convex polygons.}
  \label{fig:19}
\end{figure*}
  

  \begin{figure}[h]
    \centering
    \includegraphics[width=0.3\textwidth]{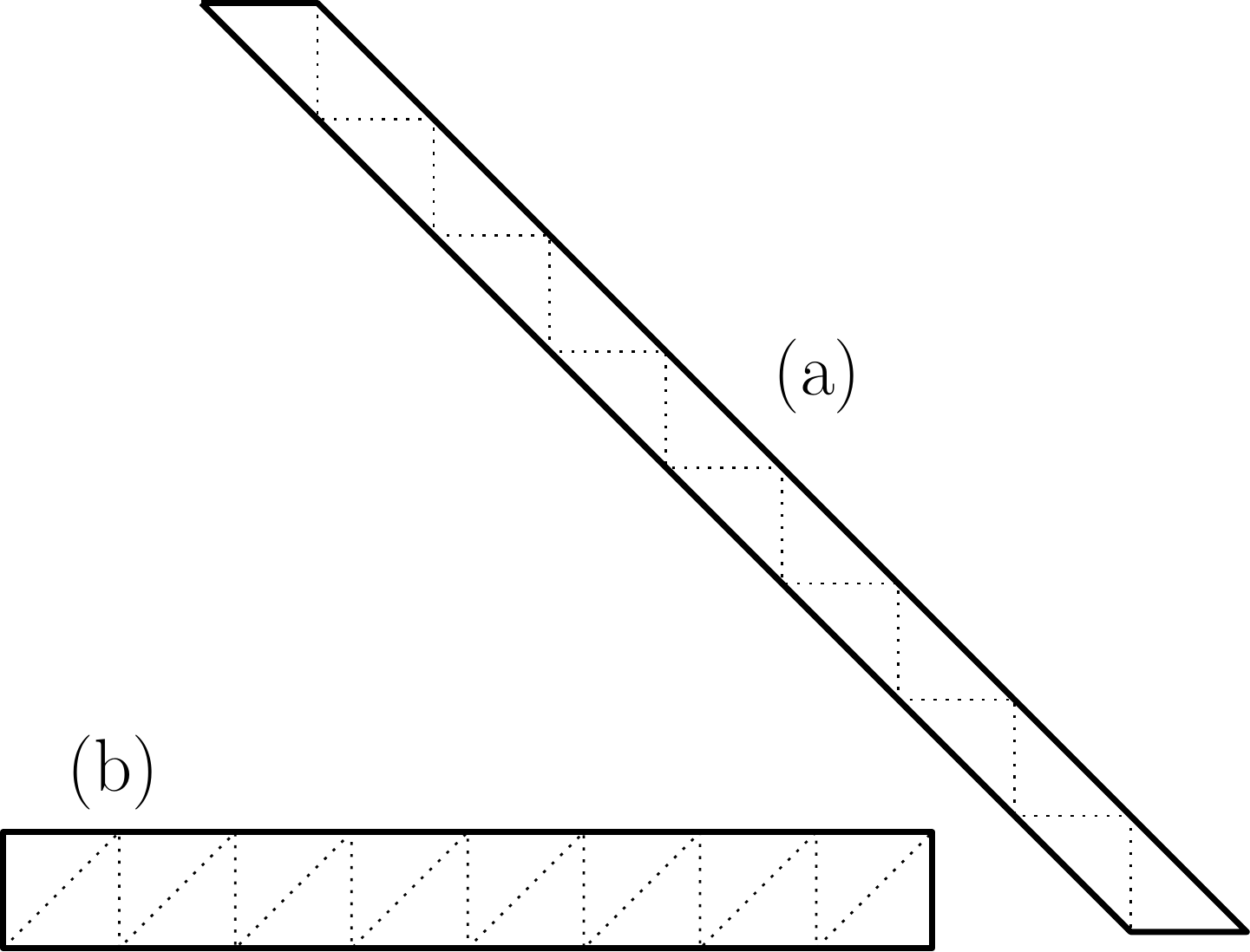}
    \caption{Any set of 7 pieces covering shape (a)
   must have a piece that consists of at least 3 triangles, 
   which cannot be covered by shape~(b).}
    \label{fig:skinny}
  \end{figure}

\section{Beyond seven pieces}

The next natural question to ask is how many pieces built from sixteen identical isosceles right triangles might one need in order to form all 20 convex polygons.

\begin{theorem}
Ten or fewer pieces formed from sixteen identical isosceles right triangles
cannot form 20 convex polygons.
However, eleven pieces can.
  \label{thm:eleven-pieces}
\end{theorem}
\begin{proof}
  In the negative direction, observe that to form the $1 \times 8\sqrt 2$ parallelogram in \figurename~\ref{fig:skinny}~(a) with ten pieces, there must be at least
six $1 \times \sqrt 2$ parallelograms and at most four single triangles (larger pieces all contain a parallelogram
  and do not fit within the shape of \figurename~\ref{fig:skinny}~(b)).

  Consider the $2\sqrt 2$-sidelength square.
  The perimeter has 8 incident triangles, so the six
  parallelograms would have to cover at least four of those.
  Exhaustive case analysis, as seen in \figurename~\ref{fig:fitting-in-square},
    shows that all arrangements that cover enough of the exterior triangles
    leave a square in the middle and boundary triangles that cannot fit a single parallelogram.

    \begin{figure}[h]
      \centering
      \includegraphics[scale=0.7]{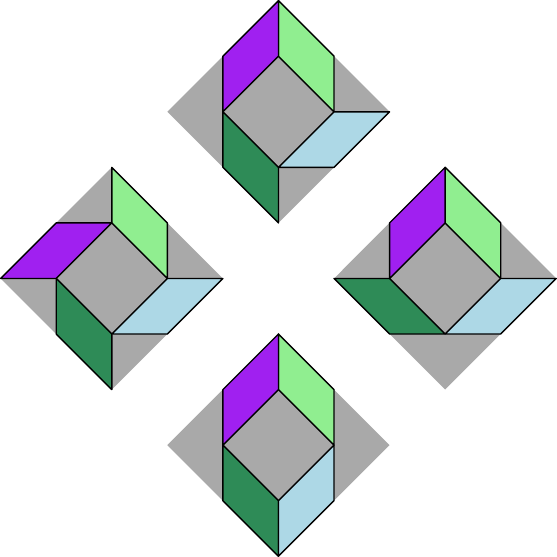}
      \caption{Six parallelograms do not fit in a square.}
      \label{fig:fitting-in-square}
    \end{figure}

  We observe that five $1 \times \sqrt 2$ parallelograms can fit inside each of the
  20 shapes: see \figurename~\ref{fig:11}.
  So, these parallelograms along with six single triangles can realize
  all 20 convex polygons. 
\end{proof}
\begin{figure*}[h]
  \centering
  \includegraphics[height=0.19\textheight]{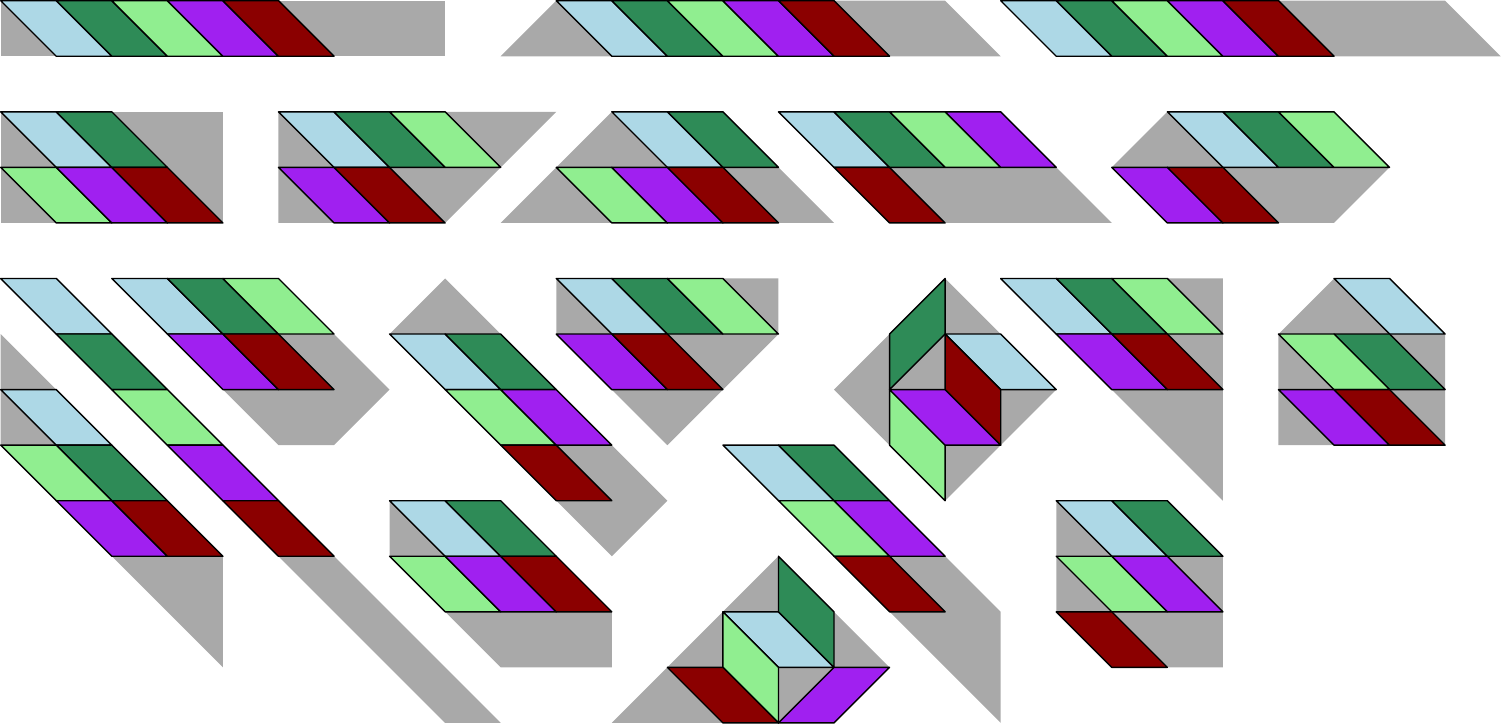}
  \caption{Eleven pieces forming all twenty convex polygons (the six individual isosceles right triangles not shown).}
  \label{fig:11}
\end{figure*}

\section{Concluding remarks}

Sixteen identical right isosceles triangles can form twenty convex polygons.
We compare the power of expression of some classic dissection puzzles constructed from
  these triangles.
The ``difficulty'' of a dissection puzzle for people to solve can be estimated by the number of ways in which one can solve it.
Computing these numbers efficiently remains a compelling task for future work.


  \begin{figure}
    \centering
    \includegraphics[width=0.48\textwidth]{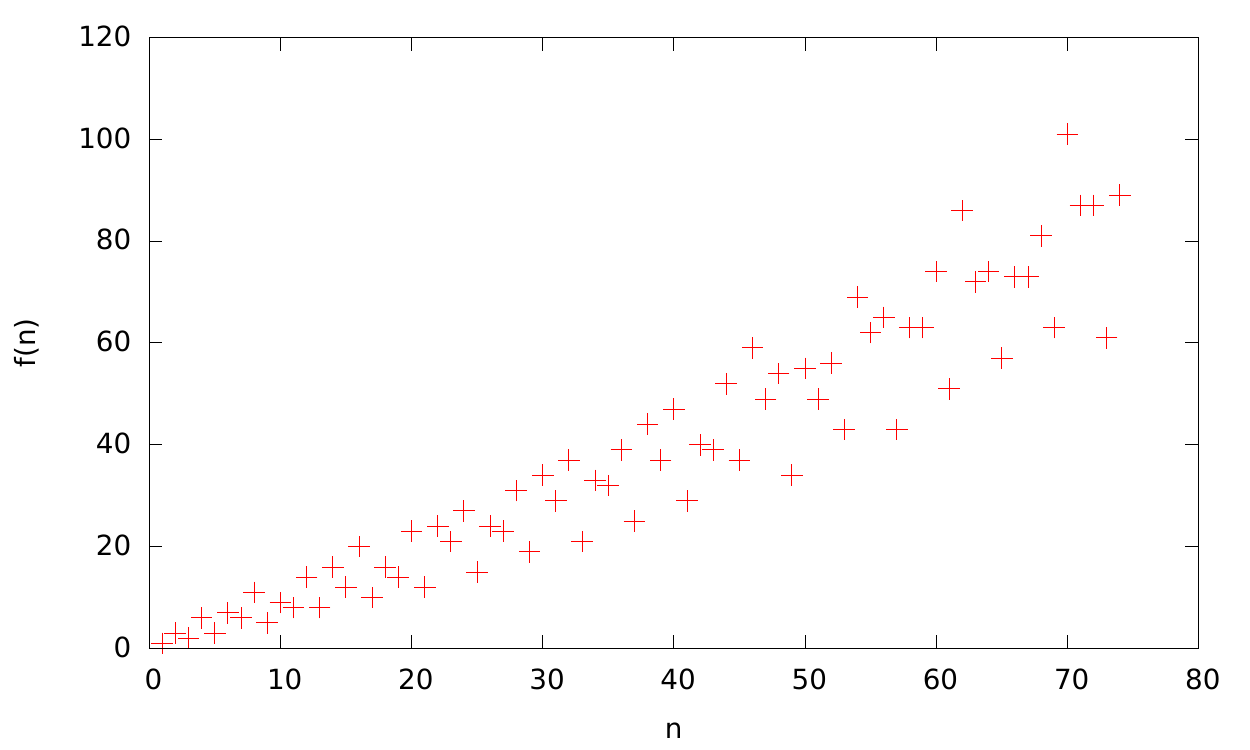}
    \caption{The number $f(n)$ of convex polygons formed by
             $n$ identical right isosceles triangles.}
    \label{fig:num}
  \end{figure}

Another interesting direction of study is the number of convex polygons formed by different numbers of triangles.
Let $f(n)$ be the number of formable convex polygons formed by $n$ identical right isosceles triangles.
To analyze the tangram and \game puzzles,
the value $f(16)=20$ plays an important role.
If we design larger puzzles,
it is natural to consider the number of formable polygons.
The function $f(n)$ itself is also interesting to investigate.
The values of the function presented in \figurename~\ref{fig:num} were determined by computer search:
all potential side length assignment to octogons were considered 
(for convex, simple arragements of identical right isosceles triangles, the interior angles are at most $\frac{3\pi}{4}$; polygons with more than 8 sides have average exterior angle strictly less than $\pi/4$).
Although it is not monotone ($f(1)=1, f(2)=3$, and $f(3)=2$),
it is a generally increasing function.
Trivially, for all $x \geq 0$, we have~$f(x) < f(2x)$ as one can subdivide every triangle into two to get the same number. The inequality's strictness comes from a new, skinnier parallelogram with side lengths 1 and $x \sqrt 2$.

\section*{Acknowledgements}
We are especially grateful to the National Science Foundation's East Asia
  and Pacific Summer Institutes and the Japan Society for the Promotion
  of Science Summer Program for encouraging international collaboration.

\bibliographystyle{alpha}
\bibliography{main}

\end{document}